\documentclass[conference]{IEEEtran}
\usepackage{cite}
\usepackage{amsmath,amssymb,amsfonts}
\usepackage{graphicx}
\usepackage{textcomp}
\usepackage{xcolor}
\def\BibTeX{{\rm B\kern-.05em{\sc i\kern-.025em b}\kern-.08em
    T\kern-.1667em\lower.7ex\hbox{E}\kern-.125emX}}
\usepackage{amsmath,amssymb,amsfonts}
\usepackage{algorithmic}
\usepackage{graphicx}
\usepackage{tikz}
\usetikzlibrary{intersections}
\usepackage[caption=false]{subfig}

\usetikzlibrary{calc}
\usepackage{bm}
\usepackage[hyphens]{url}
\usepackage[hidelinks]{hyperref}
\hypersetup{colorlinks = true, urlcolor=black}
\hypersetup{breaklinks=true}

\usepackage{multirow}
\usepackage{pgfplots}
\usepackage{textcomp}
\newtheorem{thm}{Theorem}
\newtheorem{lem}[thm]{Lemma}

\usetikzlibrary{calc,patterns,arrows,shapes.arrows,intersections}
\usetikzlibrary{decorations}
\usetikzlibrary{arrows}
\usetikzlibrary{shapes.geometric}
\usetikzlibrary{trees}
\usetikzlibrary{arrows,shapes}
\usetikzlibrary{calc,patterns,arrows,shapes.arrows,intersections}
\usetikzlibrary{arrows}
\usetikzlibrary{positioning,chains,fit,shapes,calc}
\usetikzlibrary{positioning}
\usetikzlibrary{decorations.text}
\usetikzlibrary{decorations.pathmorphing}
\usetikzlibrary{shapes.misc, positioning}
\usepackage{xcolor}
\usepackage{stackengine}
\def\delequal{\mathrel{\ensurestackMath{\stackon[1pt]{=}{\scriptstyle\Delta}}}}
\usepackage[utf8]{inputenc}
\usepackage{mathtools}
\usepackage{algorithm2e}
\newenvironment{proof}{\textit{Proof.}}

\usepgfplotslibrary{polar}
\usepackage{ellipsis}
\usetikzlibrary{calc}
\usetikzlibrary{decorations.pathreplacing,decorations.markings,shapes.geometric}

\usepackage{glossaries}
\pgfplotsset{compat=1.18}

\newacronym{iot}{IoT}{Internet of things}
\newacronym{bs}{BS}{base station}
\newacronym{miso}{MISO}{multiple-input single-output}
\newacronym{simo}{SIMO}{single-input multiple-output}
\newacronym{mimo}{MIMO}{multiple-input multiple-output}
\newacronym{siso}{SISO}{single-input single-output}
\newacronym{bps}{bps}{bits per second}
\newacronym{mrt}{MRT}{maximum-ratio transmission}
\newacronym{los}{LoS}{line-of-sight}
\newacronym{awgn}{AWGN}{additive white Gaussian noise}
\newacronym{cmos}{CMOS}{complementary metal-oxide semiconductor}
\newacronym{ecc}{ECC}{error control coding}
\newacronym{adc}{ADC}{analog-to-digital converter}
\newacronym{dac}{DAC}{digital-to-analog converter}
\newacronym{dsp}{DSP}{digital signal processing}
\newacronym{lpf}{LPF}{low-pass filter}
\newacronym{bpf}{BPF}{band-pass filter}
\newacronym{lo}{LO}{local oscillator}
\newacronym{rf}{RF}{radio frequency}
\newacronym{dc}{DC}{direct current}
\newacronym{if}{IF}{intermediate frequency}
\newacronym{pa}{PA}{power amplifier}
\newacronym{wpt}{WPT}{wireless power transfer}
\newacronym{snr}{SNR}{signal-to-noise ratio}
\newacronym{ber}{BER}{bit error rate}
\newacronym{per}{PER}{packet error rate}
\newacronym{qam}{QAM}{quadrature amplitude modulation}
\newacronym{bpsk}{BPSK}{binary phase-shift keying }
\newacronym{qpsk}{QPSK}{quadrature phase-shift keying}
\newacronym{oqpsk}{OQPSK}{offset quadrature phase-shift keying}
\newacronym{psd}{PSD}{power spectral density}
\newacronym{eirp}{EIRP}{effective isotropic radiated power}
\newacronym{ism}{ISM}{industrial, scientific and medical}
\newacronym{fcc}{FCC}{Federal Communications Commission}
\newacronym{lp}{LP}{low-power}
\newacronym{mp}{MP}{medium-power}
\newacronym{hp}{HP}{high-power}
\newacronym{cdf}{CDF}{cumulative distribution function}
\newacronym{nlos}{NLOS}{non-line-of-sight}
\newacronym{ppp}{PPP}{Poisson point process}
\newacronym{pdf}{PDF}{probability density function}
\newacronym{ris}{RIS}{reconfigurable intelligent surface}
\newacronym{csi}{CSI}{channel state information}
\newacronym{af}{AF}{amplify and forward}
\newacronym{eh}{EH}{energy harvesting}
\newacronym{em}{EM}{electromagnetic}
\newacronym{ap}{AP}{access point}
\newacronym{nim}{NIM}{negative-index material}
\newacronym{srr}{SRR}{split ring resonator}
\newacronym{iid}{i.i.d.}{independent and identically distributed}
\newacronym{ml}{ML}{maximum likelihood}
\newacronym{lse}{LSE}{least square error}
\newacronym{mse}{MSE}{mean squared error}
\newacronym{rmse}{RMSE}{root mean squared error}
\newacronym{svd}{SVD}{singular value decomposition}\newacronym{mmse}{MMSE}{minimum mean square error}
\newacronym{swipt}{SWIPT}{simultaneous wireless information and power transfer}
\IEEEoverridecommandlockouts

\begin{document}
\bstctlcite{IEEEexample:BSTcontrol}

% \title{A Phase-Alignment Algorithm for Power Maximization at Reconfigurable Intelligent Surfaces Using Amplitude Measurements}

\title{Amplitude-Based Sequential Optimization of Energy Harvesting with Reconfigurable Intelligent Surfaces\\
\thanks{This paper was supported by Digital Futures.}}

% \title{Conference Paper Title*\\

\author{\IEEEauthorblockN{Morteza~Tavana$^{*}$, Meysam~Masoudi$^{\dagger}$, and Emil~Bj\"ornson$^{*}$}\\\vspace{-2mm}
\IEEEauthorblockA{$^{*}$School of Electrical Engineering and Computer Science, KTH Royal Institute of Technology, Stockholm, Sweden\\
$^{\dagger}$Ericsson, Global AI Accelerator (GAIA) unit, Sweden\\
E-mail: $^{*}$\{morteza2, ~emilbjo\}@kth.se, $^{\dagger}$meysam.masoudi@ericsson.com}}

% \IEEEPARstart{W}{ith} the exponential growth of the number of wireless devices and with the emergence of new applications, the wireless networks need to be modified to tackle those challenges.
\maketitle

\begin{abstract}
Reconfigurable Intelligent Surfaces (RISs) have gained immense popularity in recent years because of their ability to improve wireless coverage and their flexibility to adapt to the changes in a wireless environment. These advantages are due to RISs’ ability to control and manipulate radio frequency (RF) wave propagation. RISs may be deployed in inaccessible locations where it is difficult or expensive to connect to the power grid. Energy harvesting can enable the RIS to self-sustain its operations without relying on external power sources. In this paper, we consider the problem of energy harvesting for RISs in the absence of coordination with the ambient RF source. We consider both direct and indirect energy harvesting scenarios and show that the same mathematical model applies to them. We propose a sequential phase-alignment algorithm that maximizes the received power based on only power measurements. We prove the convergence of the proposed algorithm to the optimal value {\color{black} under specific circumstances}. Our simulation results show that the proposed algorithm converges to the optimal solution in a few iterations and outperforms the random phase update method in terms of the number of required measurements.
\end{abstract}

\begin{IEEEkeywords}
Energy harvesting, reconfigurable intelligent surface, phased array, zero-energy devices.\vspace{0mm}
\end{IEEEkeywords}

\IEEEpeerreviewmaketitle
\vspace{-1mm}
\section{Introduction}
% Future wireless networks should provide resources for the rapidly growing number of devices and services. A global estimate for the number of active wireless connections is around $40$ billion by 2025 \cite{iot2025}. The energy consumption of wireless networks has driven the increasing demand for energy-saving solutions and zero-energy devices \cite{MTavana_IoTJ_2021, MTavana_2022_ICC, Mtavana_VTCW_2022}.

{\color{black} Future wireless networks should provide seamless connectivity for the rapidly growing number of devices and services. If wireless networks are implemented in the same manner as before, the energy consumption would keep increasing dramatically with the traffic volume. From both carbon footprint and energy consumption perspectives, it is a necessity to develop energy-saving techniques that can be implemented in the network nodes including low-power and zero-energy devices \cite{MTavana_IoTJ_2021, MTavana_2022_ICC, Mtavana_VTCW_2022}.}

Traditional wireless networks  have no control over the radio propagation environment.
{\color{black} Providing connectivity for regions with low \gls*{snr} comes at cost of deploying more sophisticated transmission schemes, more radio resources such as antennas and spectrum and consequently consuming more energy} \cite{goldsmith2005wireless, 9140329}.

With the emergence of the \glspl*{ris}, several limiting factors associated with the propagation environment can be eliminated. A \gls*{ris}  can manipulate the propagation environment to increase the signal strength in the desired direction and guide the electromagnetic waves toward the receiver via engineered reflections \cite{9048622,6206517,9721205,9206044,wu2021intelligent}. The \gls*{ris} is primarily envisioned for providing coverage for regions that are blocked by objects\cite{9140329}. %There are several scenarios where using \gls*{ris} improves the coverage. 
For instance, a \gls*{ris} can be deployed in a city with dense buildings and poor \gls*{los} conditions, or it can also be deployed in homes, where the walls obstruct the signal path.

In the presence of a wired power supply, \gls*{ris} may not have clear benefits compared to relays and small cells, and the respective advantages are debatable. However, if the \gls*{ris} is self-sustainable, it opens up new possibilities for deployment with no competition. This is where energy harvesting comes into play. By harvesting energy from \gls*{rf} signals that are already present in the environment, the \gls*{ris} can operate without relying on external power sources \cite{ntontin2022wireless}. This is particularly beneficial in situations where the \gls*{ris} is deployed in remote or inaccessible locations where it can be difficult or expensive to provide a continuous power source. In \cite{8741198}, the authors developed a model for the \gls*{ris} power consumption that is based on the number of elements and their phase resolution. Higher phase resolution in the \gls*{ris} increases the complexity and the power consumption \cite{7370753}.

The study \cite{9214497} considers a hybrid-relaying scheme empowered by a self-sustainable \gls*{ris} to simultaneously improve the downlink energy transmission and uplink information transmission. The authors proposed time-switching and power-splitting schemes for \gls*{ris} operation. The paper \cite{9895266} proposes a novel framework for \gls*{wpt} system using a \gls*{ris} to improve power transfer efficiency. The proposed framework employs independent beamforming to replace conventional joint beamforming, which results in higher efficiency.

The state-of-the-art techniques consider perfect \gls*{csi} at the \gls*{ris}, which is obtained by coordination between the transmitters and the \gls*{ris}. The existing solutions require extra hardware (i.e., multiple \gls*{rf} receiver to measure both amplitude and phase to obtain \gls*{csi}) and signaling, which in turn increases the energy consumption and cost of the \gls*{ris}. 

In this paper, we consider a different scenario, where the \gls*{ris} must configure itself without \gls*{rf} receivers. The proposed method makes power measurements in the energy harvesting units and uses them to maximize the harvesting power. 
\begin{figure}[t]
    \centering
    \hspace{-15mm}
    \scalebox{0.65}{\definecolor{mycolor1}{rgb}{0.85,0.325,0.098}
\definecolor{mycolor2}{rgb}{0.00000,0.44700,0.74100}
\definecolor{mycolor3}{rgb}{0.00000,0.49804,0.00000}
\definecolor{apricot}{rgb}{0.98, 0.81, 0.69}
\definecolor{copper}{rgb}{0.72, 0.45, 0.2}
\definecolor{babyblue}{rgb}{0.54, 0.81, 0.94}
\definecolor{gray}{rgb}{0.5, 0.5, 0.5}
\usetikzlibrary{decorations.pathreplacing}

\tikzset{radiation/.style={{decorate,decoration={expanding waves,angle=90,segment length=4pt}}},
         relay/.pic={
        code={\tikzset{scale=5/10}
            \draw[semithick] (0,0) -- (1,4);% left line
            \draw[semithick] (3,0) -- (2,4);% right line
            \draw[semithick] (0,0) arc (180:0:1.5 and -0.5) node[above, midway]{#1};
            \node[inner sep=4pt] (circ) at (1.5,5.5) {};
            \draw[semithick] (1.5,5.5) circle(8pt);
            \draw[semithick] (1.5,5.5cm-8pt) -- (1.5,4);
            \draw[semithick] (1.5,4) ellipse (0.5 and 0.166);
            \draw[semithick,radiation,decoration={angle=45}] (1.5cm+8pt,5.5) -- +(0:2);
            \draw[semithick,radiation,decoration={angle=45}] (1.5cm-8pt,5.5) -- +(180:2);
  }}
}

\begin{tikzpicture}[cross/.style={path picture={ 
  \draw[black]
(path picture bounding box.south east) -- (path picture bounding box.north west) (path picture bounding box.south west) -- (path picture bounding box.north east);
}}]
\pgfmathsetmacro{\R}{0.5}
\pgfmathsetmacro{\W}{1.25}
\pgfmathsetmacro{\H}{1.25}
%\tikzstyle{every node}=[trapezium, draw, minimum width=\R cm, minimum height=\R,
%trapezium left angle=60, trapezium right angle=120]
\draw[fill=gray] (-7.5,1.73) rectangle (-6,3.21) node[white] at (-6.75,2.32) {\small Controller};
\node[white] at (-6.75,2.69) {\small RIS};
\foreach \i in {1,...,14} {
    \draw[] (-6,3.23-0.1*\i)-- (-5.9,3.23-0.1*\i);
    \draw[] (-7.5,3.23-0.1*\i)-- (-7.6,3.23-0.1*\i);
    \draw[] (-6-0.1*\i,3.23)-- (-6-0.1*\i,3.33);
    \draw[] (-6-0.1*\i,1.73)-- (-6-0.1*\i,1.63);}
\draw [stealth-stealth,thick] (-5.9,2.48) -- (-3.96,2.48);
\path (-7.125,-1.7)  pic[scale=0.5,color=black] {relay={\small BS}};
\draw [stealth-stealth,thick] (-6.75,1.63) -- (-6.75,0);
\filldraw[fill=babyblue, draw=babyblue] (-3.14-0.82,-1.5+0.82) rectangle (0.88-0.82,2.52+0.82) node[] at (-1.2-0.82,2.15+0.82) {\small Control circuit board};
\filldraw[fill=copper, draw=copper] (-3.14,-1.5) rectangle (0.88,2.52) node[] at (-1.2,2.15) {\small Copper};
\filldraw[fill=apricot, draw=apricot] (-2.32,-2.32) rectangle (1.7,1.7);
\foreach \i in {-3,...,2} {
    \foreach \j in {-3,...,2} {
    \node [draw, copper, line width=2,trapezium, trapezium left angle=90, trapezium right angle=90, minimum width=\R cm,outer sep=0pt,fill=apricot, rotate = 0] at (\i*\W*\R,\j*\H*\R) {};
    \draw[copper, line width=2] (\i*\W*\R,\j*\H*\R+0.5*\R) -- (\i*\W*\R,\j*\H*\R+0.15*\R);
    \draw[copper, line width=2] (\i*\W*\R,\j*\H*\R-0.5*\R) -- (\i*\W*\R,\j*\H*\R-0.15*\R);
    \draw[copper, line width=2] (\i*\W*\R-0.25*\R,\j*\H*\R-0.15*\R) -- (\i*\W*\R+0.25*\R,\j*\H*\R-0.15*\R);
    \draw[copper, line width=2] (\i*\W*\R-0.25*\R,\j*\H*\R+0.15*\R) -- (\i*\W*\R+0.25*\R,\j*\H*\R+0.15*\R);}}
    
\begin{polaraxis}[at={(-3.87in,-1.65in)},
		grid=none,
		width=9cm,
		xtick=\empty,
		ytick=\empty,
		rotate=5,
		axis line style={draw=none},
		tick style={draw=none}]
		\addplot+[mark=none, fill=mycolor3, opacity=0.4, line width=0.75pt, color = mycolor3, domain=-180:180,samples=600] 
		{0.8*abs(sin(2*x*3.141592)/(3.141592*3.141592/180*x))}; 
\end{polaraxis}

\begin{polaraxis}[at={(-1.9in,-1.42in)},
		grid=none,
		width=9cm,
		xtick=\empty,
		ytick=\empty,
		rotate=30,
		axis line style={draw=none},
		tick style={draw=none}]
		\addplot+[mark=none, fill=mycolor3, opacity=0.4, line width=0.75pt, color = mycolor3, domain=-180:180,samples=600] 
		{0.1*abs(sin(4*x*3.141592)/(3.141592*3.141592/180*x))}; 
\end{polaraxis}

\begin{polaraxis}[at={(-1.84in,-1.65in)},
		grid=none,
		width=9cm,
		xtick=\empty,
		ytick=\empty,
		rotate=-15,
		axis line style={draw=none},
		tick style={draw=none}]
		\addplot+[mark=none, fill=mycolor3, opacity=0.4, line width=0.75pt, color = mycolor3, domain=-180:180,samples=600] 
		{0.1*abs(sin(4*x*3.141592)/(3.141592*3.141592/180*x))}; 
\end{polaraxis}

\begin{polaraxis}[at={(-1.89in,-1.94in)},
		grid=none,
		width=9cm,
		xtick=\empty,
		ytick=\empty,
		rotate=-60,
		axis line style={draw=none},
		tick style={draw=none}]
		\addplot+[mark=none, fill=mycolor3, opacity=0.4, line width=0.75pt, color = mycolor3, domain=-180:180,samples=600] 
		{0.1*abs(sin(4*x*3.141592)/(3.141592*3.141592/180*x))}; 
\end{polaraxis}

\end{tikzpicture}}
    \vspace{-10mm}
    \caption{\small Hardware structure of a RIS.}
    \label{fig:RIS_structure}
\end{figure}
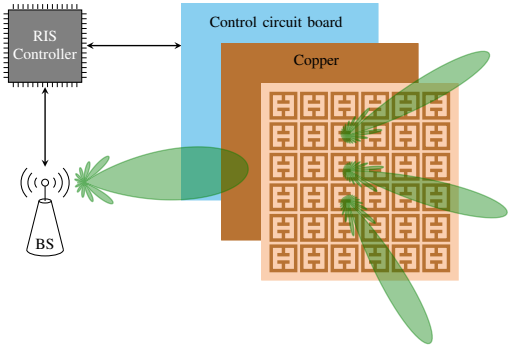
The main contributions of this paper as compared to the existing works are as follows.

\begin{itemize}
    \item We propose a sequential phase-alignment algorithm to maximize the received power at the harvesting unit based on power measurements.
    \item We prove the optimality of the proposed algorithm analytically.
    \item Our simulation results show that the proposed algorithm greatly outperforms the benchmark random phase-update method in terms of the number of required measurements to achieve the optimum.
\end{itemize}

The remainder of this paper is organized as follows. Section~\ref{sec:PF} describes the \gls*{ris} hardware architecture, two different energy harvesting schemes in the \gls*{ris}, and the phase alignment problem. Section \ref{sec:PS} describes the proposed abstract model of the \gls*{ris} operation and the proposed solution. Simulation results are presented in Section \ref{sec:NumResults}, while Section \ref{sec:Conclusion} provides our conclusions.

\textit{Notations}: We denote sets by upper-case script letters, e.g., $\mathcal{S}$. The only exceptions are the sets of natural numbers, real numbers, and complex numbers that we  represent with $\mathbb{N}$, $\mathbb{R}$, and $\mathbb{C}$, respectively.  The cardinality of a set $\mathcal{S}$ is represented by $|\mathcal{S}|$. Vectors are indicated by lower-case bold-face letters, e.g., $\bm{x}$, and $x_{i}$ denotes the $i$th element of $\bm{x}$. We represent matrices by upper-case bold-face letters, e.g., $\bm{A}$. We also use $\delequal$ to indicate an equal by definition sign. With $\left(\cdot\right)^{\mathsf{T}}$, we denote the transpose operator. We denote the imaginary unit by $j\delequal\sqrt{-1}$. We represent the conjugate of a complex number $z$ with $z^{*}$. However, we denote the optimal solution with the superscript ${\star}$, e.g., $x^{\star}$. Also, the operation $\operatorname{Arg}{\left(z\right)}$ returns the single-valued argument of $z$ that lies within the interval $(-\pi, \pi]$, while $\arg{\left(z\right)}$ returns the set of all possible values of the argument of $z$.

\section{Problem Description}\label{sec:PF}
% This section presents a simplified schematic illustration of the operation of an \gls*{ris} connected to an energy harvesting device. Following that, two distinct energy harvesting schemes are described. Lastly, the phase alignment problem for \gls*{ris}-assisted energy harvesting is described.
This section presents 1) the \gls*{ris} device architecture, 2) the energy harvesting schemes, and 3) the phase alignment problem for \gls*{ris}-assisted energy harvesting.

\subsection{RIS Hardware Architecture}

Fig. \ref{fig:RIS_structure} illustrates the hardware architecture of a typical \gls*{ris}. The front layer consists of several metal elements that are printed on a dielectric substrate and arranged in a two-dimensional array to reflect the incoming signals. For each patch element, there is a controllable circuit that is used to adjust the phase of the reflected signals and direct the signal in a desired direction.
 There is a control circuit board that can control the reflection amplitudes and phase shifts of the patch elements. The \gls*{ris} controller can communicate with the network components such as \glspl*{bs} via a connectivity interface \cite{9140329}.  Finally, the \gls*{ris} requires a power supply to adjust the phases and then maintain the desired reflection state.

\subsection{Energy Harvesting Schemes}
In general, an \gls*{ris} can harvest \gls*{rf} energy directly or indirectly from an ambient \gls*{rf} source. We will describe these scenarios below and later show that they lead to system models of the same kind. 

\subsubsection{Direct Energy Harvesting}
The \gls*{ris} elements are capable of reflecting and receiving the incident \gls*{em} waves. During the harvesting phase, the \gls*{ris} operates in the reception mode, where it combines the received signals from each element with some phase shifts.\footnote{There is a type of metasurface implementation called holographic beamforming that can pass the incident EM waves from one side to the other. The metamaterial can add adjustable phase shifts, similar to a phased array, but with a different implementation \cite{black2017holographic}.} The \gls*{ris} can use the harvested energy to fully sustain its operation or, if the energy is insufficient, it can decrease the consumption from other energy sources, as a first step toward achieving a zero-energy \gls*{ris} system. The abstract model of the direct energy harvesting operation is shown in Fig. \ref{fig:direct_RIS}.

\subsubsection{Indirect Energy Harvesting}
The \gls*{ris} is generally deployed and designed to reflect the \gls*{em} waves towards a desired location (typically a receiver location). Inspired by this principle, the energy harvesting device can alternatively be deployed in front of the \gls*{ris} (at a short distance) and the phases of the \gls*{ris} elements can be aligned so they combine constructively at the location of the energy harvesting device. This device then can return the energy to the \gls*{ris} via a cable. The concept of indirect energy harvesting is demonstrated in Fig. \ref{fig:Indirect_RIS}.

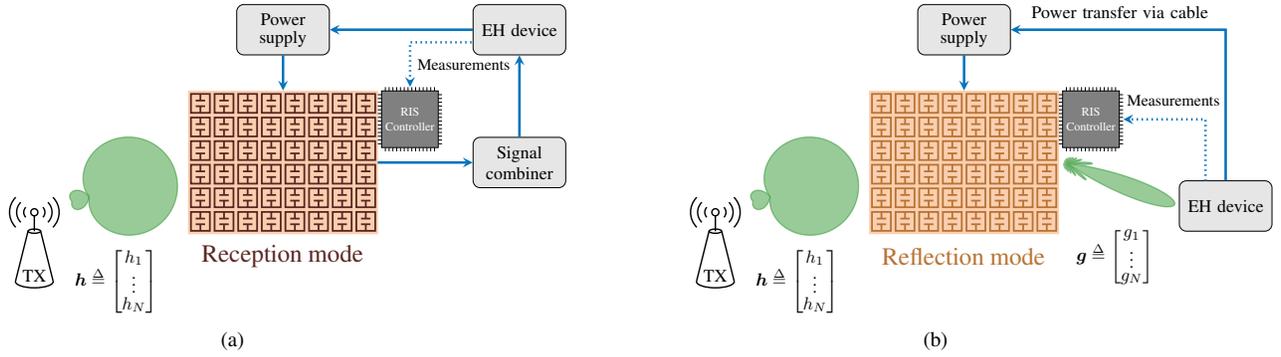
\begin{figure*}[!t]
          \centering
          \hspace{-17mm}
          \subfloat[]{\scalebox{0.67}{\definecolor{mycolor1}{rgb}{0.85,0.325,0.098}
\definecolor{mycolor2}{rgb}{0.00000,0.44700,0.74100}
\definecolor{mycolor3}{rgb}{0.00000,0.49804,0.00000}
\definecolor{carnationpink}{rgb}{1.0, 0.65, 0.79}
\definecolor{carnelian}{rgb}{0.7, 0.11, 0.11}
\definecolor{caputmortuum}{rgb}{0.35, 0.15, 0.13}
\definecolor{apricot}{rgb}{0.98, 0.81, 0.69}
\definecolor{copper}{rgb}{0.72, 0.45, 0.2}
\definecolor{babyblue}{rgb}{0.54, 0.81, 0.94}
\definecolor{gray}{rgb}{0.5, 0.5, 0.5}
\usetikzlibrary{decorations.pathreplacing}
\linespread{0.6}
\tikzstyle{startstop} = [rectangle, rounded corners, minimum width=1.5cm, minimum height=1cm,text centered, draw=black, fill=gray!20, inner sep=-0.3ex]

\tikzstyle{arrow} = [thick,->,>=stealth]

\begin{tikzpicture}[cross/.style={path picture={ 
  \draw[black]
(path picture bounding box.south east) -- (path picture bounding box.north west) (path picture bounding box.south west) -- (path picture bounding box.north east);
}}]

\tikzset{radiation/.style={{decorate,decoration={expanding waves,angle=90,segment length=4pt}}},
         relay/.pic={
        code={\tikzset{scale=5/10}
            \draw[thick] (0,0) -- (1,4);% left line
            \draw[thick] (3,0) -- (2,4);% right line
            \draw[thick] (0,0) arc (180:0:1.5 and -0.5) node[above, midway]{#1};
            \node[inner sep=4pt] (circ) at (1.5,5.5) {};
            \draw[thick] (1.5,5.5) circle(8pt);
            \draw[thick] (1.5,5.5cm-8pt) -- (1.5,4);
            \draw[thick] (1.5,4) ellipse (0.5 and 0.166);
            \draw[thick,radiation,decoration={angle=45}] (1.5cm+8pt,5.5) -- +(0:2);
            \draw[thick,radiation,decoration={angle=45}] (1.5cm-8pt,5.5) -- +(180:2);
  }}
}

\pgfmathsetmacro{\R}{0.375}
\pgfmathsetmacro{\W}{1.25}
\pgfmathsetmacro{\H}{1.25}
\pgfmathsetmacro{\RISx}{-5.455}
\pgfmathsetmacro{\RISy}{3.6}

\filldraw[fill=apricot, draw=apricot] (\RISx-2.11,\RISy-1.65) rectangle (\RISx+1.64,\RISy+1.18);
\foreach \i in {-4,...,3} {
    \foreach \j in {-3,...,2} {
    \node [draw, caputmortuum, line width=1,trapezium, trapezium left angle=90, trapezium right angle=90, minimum width=\R cm,outer sep = 0pt,fill=apricot, rotate = 0] at (\RISx+\i*\W*\R,\RISy+\j*\H*\R) {};
    \draw[caputmortuum, line width=1] (\RISx+\i*\W*\R,\RISy+\j*\H*\R+0.5*\R) -- (\RISx+\i*\W*\R,\RISy+\j*\H*\R+0.15*\R);
    \draw[caputmortuum, line width=1] (\RISx+\i*\W*\R,\RISy+\j*\H*\R-0.5*\R) -- (\RISx+\i*\W*\R,\RISy+\j*\H*\R-0.15*\R);
    \draw[caputmortuum, line width=1] (\RISx+\i*\W*\R-0.25*\R,\RISy+\j*\H*\R-0.15*\R) -- (\RISx+\i*\W*\R+0.25*\R,\RISy+\j*\H*\R-0.15*\R);
    \draw[caputmortuum, line width=1] (\RISx+\i*\W*\R-0.25*\R,\RISy+\j*\H*\R+0.15*\R) -- (\RISx+\i*\W*\R+0.25*\R,\RISy+\j*\H*\R+0.15*\R);}}

\path (-11,1)  pic[scale=0.5,color=black] {relay={TX}};

\node[text width=5.5em]  (SC)  at (-1,3.363) [startstop] {Signal\\\vspace{1mm}combiner};

\node[text width=5.5em]  (PS)  at (-5.69,6) [startstop] {Power\\\vspace{1mm}supply};

\node[text width=5.5em]  (EH)  at (-1,6) [startstop] {EH device};

\draw[arrow, line width=0.5mm, mycolor2] (-3.812,3.363) -- (SC.west);

\draw[arrow, line width=0.5mm, mycolor2] (SC.north) -- (EH.south);
\draw[arrow, line width=0.5mm, mycolor2] (EH.west) -- (PS.east);

\draw[arrow, line width=0.5mm, mycolor2] (PS.south) --++ (0,-0.7);

\node[caputmortuum] at (-5.7, 1.5) {\Large Reception mode};

\node[black] at (-9, 1) {$\bm{h}\delequal \begin{bmatrix}
h_{1} \\
\vdots \\
h_{N}
\end{bmatrix}$};

% \node[black] at (-5.2, 0.25) {Phase-shifts: $\bm{\omega}_{\theta}\delequal \begin{bmatrix}
% e^{j\vartheta_{1}} \\
% \vdots \\
% e^{j\vartheta_{N}}
% \end{bmatrix}$};

\node[black] at (-2.1, 5.3) {\small Measurements};

\draw[arrow, dotted, line width=0.5mm, mycolor2] (-1.95, 5.75) --++ (-1.22,0) --++ (0,-0.885);

% \pgfmathsetmacro{\ControllerX}{-5.455}
% \pgfmathsetmacro{\ControllerY}{5.5}

\pgfmathsetmacro{\ControllerX}{2.52}
\pgfmathsetmacro{\ControllerY}{3.15}
\begin{pgflowlevelscope}{\pgftransformscale{0.75}}
\draw[fill=gray] (\ControllerX-7.5,\ControllerY+1.73) rectangle (\ControllerX-6,\ControllerY+3.23) node[white] at (\ControllerX-6.75,\ControllerY+2.30) {\small Controller};
\node[white] at (\ControllerX-6.75,\ControllerY+2.69) {\small RIS};
\foreach \i in {1,...,14} {
    \draw[] (\ControllerX-6,\ControllerY+3.23-0.1*\i)-- (\ControllerX-5.9,\ControllerY+3.23-0.1*\i);
    \draw[] (\ControllerX-7.5,\ControllerY+3.23-0.1*\i)-- (\ControllerX-7.6,\ControllerY+3.23-0.1*\i);
    \draw[] (\ControllerX-6-0.1*\i,\ControllerY+3.23)-- (\ControllerX-6-0.1*\i,\ControllerY+3.33);
    \draw[] (\ControllerX-6-0.1*\i,\ControllerY+1.73)-- (\ControllerX-6-0.1*\i,\ControllerY+1.63);}
\end{pgflowlevelscope}

\begin{polaraxis}[at={(-4.61in,0.2in)},
		grid=none,
		width=5.5cm,
		xtick=\empty,
		ytick=\empty,
		rotate=15,
		axis line style={draw=none},
		tick style={draw=none}]
		\addplot+[mark=none, fill=mycolor3, opacity=0.4, line width=0.75pt, color = mycolor3, domain=-180:180,samples=600] 
		{0.1*abs(sin(0.5*x*3.141592)/(3.141592*3.141592/180*x))}; 
\end{polaraxis}

\end{tikzpicture}}%
              \label{fig:direct_RIS}}
          \hfil
          \hspace{-5mm}
          \subfloat[]{\scalebox{0.67}{\definecolor{mycolor1}{rgb}{0.85,0.325,0.098}
\definecolor{mycolor2}{rgb}{0.00000,0.44700,0.74100}
\definecolor{mycolor3}{rgb}{0.00000,0.49804,0.00000}

\definecolor{apricot}{rgb}{0.98, 0.81, 0.69}
\definecolor{copper}{rgb}{0.72, 0.45, 0.2}
\definecolor{babyblue}{rgb}{0.54, 0.81, 0.94}
\definecolor{gray}{rgb}{0.5, 0.5, 0.5}
\usetikzlibrary{decorations.pathreplacing}
\linespread{0.6}
\tikzstyle{startstop} = [rectangle, rounded corners, minimum width=1.5cm, minimum height=1cm,text centered, draw=black, fill=gray!20, inner sep=-0.3ex]

\tikzstyle{arrow} = [thick,->,>=stealth]

\begin{tikzpicture}[cross/.style={path picture={ 
  \draw[black]
(path picture bounding box.south east) -- (path picture bounding box.north west) (path picture bounding box.south west) -- (path picture bounding box.north east);
}}]

\tikzset{radiation/.style={{decorate,decoration={expanding waves,angle=90,segment length=4pt}}},
         relay/.pic={
        code={\tikzset{scale=5/10}
            \draw[thick] (0,0) -- (1,4);% left line
            \draw[thick] (3,0) -- (2,4);% right line
            \draw[thick] (0,0) arc (180:0:1.5 and -0.5) node[above, midway]{#1};
            \node[inner sep=4pt] (circ) at (1.5,5.5) {};
            \draw[thick] (1.5,5.5) circle(8pt);
            \draw[thick] (1.5,5.5cm-8pt) -- (1.5,4);
            \draw[thick] (1.5,4) ellipse (0.5 and 0.166);
            \draw[thick,radiation,decoration={angle=45}] (1.5cm+8pt,5.5) -- +(0:2);
            \draw[thick,radiation,decoration={angle=45}] (1.5cm-8pt,5.5) -- +(180:2);
  }}
}

\pgfmathsetmacro{\R}{0.375}
\pgfmathsetmacro{\W}{1.25}
\pgfmathsetmacro{\H}{1.25}
\pgfmathsetmacro{\RISx}{-5.455}
\pgfmathsetmacro{\RISy}{3.6}

\filldraw[fill=apricot, draw=apricot] (\RISx-2.11,\RISy-1.65) rectangle (\RISx+1.64,\RISy+1.18);
\foreach \i in {-4,...,3} {
    \foreach \j in {-3,...,2} {
    \node [draw, copper, line width=1,trapezium, trapezium left angle=90, trapezium right angle=90, minimum width=\R cm,outer sep = 0pt,fill=apricot, rotate = 0] at (\RISx+\i*\W*\R,\RISy+\j*\H*\R) {};
    \draw[copper, line width=1] (\RISx+\i*\W*\R,\RISy+\j*\H*\R+0.5*\R) -- (\RISx+\i*\W*\R,\RISy+\j*\H*\R+0.15*\R);
    \draw[copper, line width=1] (\RISx+\i*\W*\R,\RISy+\j*\H*\R-0.5*\R) -- (\RISx+\i*\W*\R,\RISy+\j*\H*\R-0.15*\R);
    \draw[copper, line width=1] (\RISx+\i*\W*\R-0.25*\R,\RISy+\j*\H*\R-0.15*\R) -- (\RISx+\i*\W*\R+0.25*\R,\RISy+\j*\H*\R-0.15*\R);
    \draw[copper, line width=1] (\RISx+\i*\W*\R-0.25*\R,\RISy+\j*\H*\R+0.15*\R) -- (\RISx+\i*\W*\R+0.25*\R,\RISy+\j*\H*\R+0.15*\R);}}

\path (-11,1)  pic[scale=0.5,color=black] {relay={TX}};

\node[text width=5.5em]  (EH)  at (-0.5,2.5) [startstop] {EH device};

\node at (-2.6, 6.35) {Power transfer via cable};

\node[copper] at (-5.7, 1.5) {\Large Reflection mode};

\node[text width=5.5em]  (PS)  at (-5.69,6) [startstop] {Power\\\vspace{1mm}supply};

\draw[arrow, line width=0.5mm, mycolor2] (EH.north) |- (PS.east);

\draw[arrow, line width=0.5mm, mycolor2] (PS.south) --++ (0,-0.7);

\node[black] at (-9, 1) {$\bm{h}\delequal \begin{bmatrix}
h_{1} \\
\vdots \\
h_{N}
\end{bmatrix}$};

\node[black] at (-2.7, 1.5) {$\bm{g}\delequal \begin{bmatrix}
g_{1} \\
\vdots \\
g_{N}
\end{bmatrix}$};

% \node[black] at (-5.2, 0.25) {Phase-shifts: $\bm{\omega}_{\theta}\delequal \begin{bmatrix}
% e^{j\vartheta_{1}} \\
% \vdots \\
% e^{j\vartheta_{N}}
% \end{bmatrix}$};

\node[black] at (-1.53, 4.6) {\small Measurements};

% \draw[arrow, dotted, line width=0.5mm, mycolor2] (-0.95, 3) --++ (0, 2.65) --++ (-1.71,0) --++ (0,-0.52);

\draw[arrow, dotted, line width=0.5mm, mycolor2] (-0.9, 3) --++ (0, 1.22) --++ (-1.63,0);

% \pgfmathsetmacro{\ControllerX}{-5.455}
% \pgfmathsetmacro{\ControllerY}{5.5}

\pgfmathsetmacro{\ControllerX}{2.52}
\pgfmathsetmacro{\ControllerY}{3.15}
\begin{pgflowlevelscope}{\pgftransformscale{0.75}}
\draw[fill=gray] (\ControllerX-7.5,\ControllerY+1.73) rectangle (\ControllerX-6,\ControllerY+3.23) node[white] at (\ControllerX-6.75,\ControllerY+2.30) {\small Controller};
\node[white] at (\ControllerX-6.75,\ControllerY+2.69) {\small RIS};
\foreach \i in {1,...,14} {
    \draw[] (\ControllerX-6,\ControllerY+3.23-0.1*\i)-- (\ControllerX-5.9,\ControllerY+3.23-0.1*\i);
    \draw[] (\ControllerX-7.5,\ControllerY+3.23-0.1*\i)-- (\ControllerX-7.6,\ControllerY+3.23-0.1*\i);
    \draw[] (\ControllerX-6-0.1*\i,\ControllerY+3.23)-- (\ControllerX-6-0.1*\i,\ControllerY+3.33);
    \draw[] (\ControllerX-6-0.1*\i,\ControllerY+1.73)-- (\ControllerX-6-0.1*\i,\ControllerY+1.63);}
\end{pgflowlevelscope}

\begin{polaraxis}[at={(-2.55in,0.2in)},
		grid=none,
		width=6.6cm,
		xtick=\empty,
		ytick=\empty,
		rotate=-21,
		axis line style={draw=none},
		tick style={draw=none}]
		\addplot+[mark=none, fill=mycolor3, opacity=0.4, line width=0.75pt, color = mycolor3, domain=-180:180,samples=600] 
		{0.08*abs(sin(4*x*3.141592)/(3.141592*3.141592/180*x))}; 
\end{polaraxis}

\begin{polaraxis}[at={(-4.61in,0.2in)},
		grid=none,
		width=5.5cm,
		xtick=\empty,
		ytick=\empty,
		rotate=15,
		axis line style={draw=none},
		tick style={draw=none}]
		\addplot+[mark=none, fill=mycolor3, opacity=0.4, line width=0.75pt, color = mycolor3, domain=-180:180,samples=600] 
		{0.1*abs(sin(0.5*x*3.141592)/(3.141592*3.141592/180*x))}; 
\end{polaraxis}

\end{tikzpicture}}%
             \label{fig:Indirect_RIS}}
          \caption{Energy harvesting schemes: (a) direct (b) indirect.}
          \label{fig:direct_Indirect_RIS}
\end{figure*}

\subsection{Phase Alignment Problem for the Energy Harvesting}
We consider a scenario of \gls*{eh} from an ambient \gls*{rf} source by an \gls*{ris} with no prior  \gls*{csi}, and there is no coordination between the \gls*{ris} and the \gls*{rf} transmitter. We assume that the transmit power and the location of the transmitter are unknown.\footnote{This scenario is more practical (compared to a scenario with coordination between the transmitter and the \gls*{ris}) as the transmitter may not be designed to coordinate the \gls*{csi} with the \gls*{ris} or only do it when it requests that the \gls*{ris} is supporting its data transmissions.}

For the direct (indirect) scheme, we assume the measured \gls*{rf} power by the \gls*{ris} (\gls*{eh} device) has the following expression \cite{9721205}
\begin{equation}\label{eq:Power_func}
    y =\left|\sum_{n=1}^{N} z_{n}e^{j\vartheta_{n}}\right|^{2},
\end{equation}
where $z_{n}\in\mathbb{C}$ for each $0\leq n\leq N$. In \eqref{eq:Power_func}, $z_{n}$ not only includes all channel gains between transmitters and the energy harvester (except the adjustable phase shift $\vartheta_{n}$), but also takes into account the transmission power. Also, $\vartheta_{n}\in[0, 2\pi)$ is the adjustable phase shift of the $n$th element of the \gls*{ris}. Also $y$ represents the measured received power that is obtained from the reading of the power in the \gls*{eh} device\cite{9721205,haykin2008communication}.\footnote{Power measurements can be obtained from the input of the energy harvesting unit using a circuit such as a voltmeter. Alternatively, power measurements can be applied at the output of the energy harvesting unit by compensating for the nonlinear harvesting conversion efficiency.} 
% \footnote{In general, energy harvesting devices become nonlinear in high power regimes. However, power measurements can be done by nonlinearity compensation.}

For direct and indirect energy harvesting schemes, and for all $1\leq n\leq N$, we have
\begin{equation} \label{eq:z_definition}
    z_{n}\delequal\begin{cases}
    P_{\text{t}}h_{n}, \qquad &\text{Direct EH}\\
    P_{\text{t}}h_{n}g_{n}, \qquad &\text{Indirect EH},
    \end{cases}
\end{equation}
where $h_{n}\in \mathbb{C}$ and $g_{n}\in \mathbb{C}$ are the channel gains from the transmitter to the \gls*{ris} element $n$ and from the \gls*{ris} element $n$ to the \gls*{eh} device in the indirect case, respectively. Also, $P_{\text{t}}$ is the transmit power.

We assume no \gls*{csi} is available at the \gls*{ris} controller (i.e., $z_{n}$ is unavailable to the \gls*{ris}). This scenario is more general compared to a scenario with coordination between the transmitters and the \gls*{ris}, since the transmitter might lack a protocol to reach the \gls*{ris} and can even be unaware of its existence.

The values of $\left\{z_{n}\right\}_{n=1}^{N}$ depend on the geometry and propagation environment and can be estimated by measuring amplitude and phase using \gls*{rf} receiver circuits. Since the \gls*{ris} lack such \gls*{rf} chains, the \gls*{ris} cannot estimate the amplitudes and phases. Hence, the values of $\left\{z_{n}\right\}_{n=1}^N$ will remain unknown to the \gls*{ris}. On the other hand, with power measurement at the \gls*{eh} device, the \gls*{ris} can measure the combined power from all elements for any feasible phase configuration. 

The general optimization problem is to maximize the received \gls*{rf} power\footnote{Note that the harvested \gls*{rf} power is a nonlinear function of the received \gls*{rf} power that is called the conversion efficiency function. We consider the problem of choosing optimal phase shift $\bm{\vartheta}^{\star}$ for the \gls*{ris} elements to maximize the harvesting \gls*{rf} power. However, since the conversion efficiency function is generally an increasing function, the optimal solution for maximization of the harvesting \gls*{rf} power is equivalent to the one for the maximization of the measured received \gls*{rf} power.} by finding proper phase shifts (i.e.,  $\bm{\vartheta}{\delequal}\left[\vartheta_{1},\dots,\vartheta_{N}\right]^{\mathsf{T}}$) for the \gls*{ris} array.

\section{Proposed Phase Alignment Scheme}\label{sec:PS}
In this section, we investigate the problem of received \gls*{rf} power maximization using a dynamic sequence of power measurements. We propose a model for the RIS energy harvesting operations and an algorithm to find the optimal phase of \gls*{ris} elements. %In this section, we assume that the \gls*{ris} elements can add any continuous phase shifts to the incident \gls*{rf} signal.

\subsection{An Abstract Model of the Proposed RIS Operation}
A simplified model of the operation of an \gls*{ris} with an energy harvesting module is shown in Fig. \ref{fig:Harvesting_Structure}. The network entity manager can assess the network environment, including but not limited to network demand, \gls*{snr}, and power measurements from the \gls*{ris}, and determine if the \gls*{ris} should function to provide connectivity or to harvest energy. Within the \gls*{ris} controller, the preferred phase of each element can be determined based on functionality, and the controller can adjust the phase of the elements accordingly. The energy harvester can be located either outside or inside the \gls*{ris} surface. It sends power measurements to a phase alignment algorithm implemented in the network entity manager module. These values can be used by the network entity manager to decide whether the \gls*{ris} should be in energy harvesting or data transmission mode. After reaching an appropriate phase configuration, the harvested energy can be fed back to the power supply to be used by the \gls*{ris}.

\begin{figure}[t]
    \hspace{-9mm}
    \scalebox{0.8}{% \definecolor{azure}{rgb}{0.0, 0.5, 1.0}
% \definecolor{awesome}{rgb}{1.0, 0.13, 0.32}
% \definecolor{amber}{rgb}{1.0, 0.49, 0.0}
%\definecolor{darkpastelgreen}{rgb}{0.01, 0.75, 0.24}
\definecolor{mycolor1}{rgb}{0.85,0.325,0.098}
\definecolor{mycolor2}{rgb}{0.00000,0.44700,0.74100}
\definecolor{mycolor3}{rgb}{0.00000,0.49804,0.00000}

\definecolor{apricot}{rgb}{0.98, 0.81, 0.69}
\definecolor{copper}{rgb}{0.72, 0.45, 0.2}
\definecolor{babyblue}{rgb}{0.54, 0.81, 0.94}
\definecolor{gray}{rgb}{0.5, 0.5, 0.5}
\usetikzlibrary{decorations.pathreplacing}
\linespread{0.6}
\tikzstyle{startstop} = [rectangle, rounded corners, minimum width=1.5cm, minimum height=0.8cm,text centered, draw=black, fill=mycolor1!20, inner sep=-0.3ex]
\tikzstyle{startstop_blue} = [rectangle, rounded corners, minimum width=1.5cm, minimum height=0.8cm,text centered, draw=black, fill=mycolor2!20, inner sep=-0.3ex]
\tikzstyle{startstop_green} = [rectangle, rounded corners, minimum width=1.5cm, minimum height=0.8cm,text centered, draw=black, fill=mycolor3!20, inner sep=-0.3ex]

% \tikzstyle{io} = [trapezium, trapezium left angle=70, trapezium right angle=110, minimum width=3cm, minimum height=1cm, text centered, draw=black, fill=blue!30]
%\tikzstyle{process} = [rectangle, minimum width=3cm, minimum height=0.8cm, text centered, draw=black, fill=orange!30, inner sep=-0.3ex]
\tikzstyle{decision} = [diamond, aspect=2,  text centered, draw=black, fill=mycolor1!20, inner sep=-0.3ex]
\tikzstyle{arrow} = [thick,->,>=stealth]

\begin{tikzpicture}[cross/.style={path picture={ 
  \draw[black]
(path picture bounding box.south east) -- (path picture bounding box.north west) (path picture bounding box.south west) -- (path picture bounding box.north east);
}}]

%\node (in1) [io, below of=start] {Input};

%\node (pro1) [process, below of=in1] {Process 1};

\node[text width=5.5em]  (NE) [startstop] {\scriptsize Network entity\\ manager};
\node[text width=4.5em] (OM) [decision, above of=NE, yshift = 3cm] {\scriptsize Operational mode};
\node[text width=8.5em] (PA1) [startstop_blue, left of=OM, yshift = 1cm, xshift = -2.3cm] {\scriptsize Phase alignment algorithm \\ for data transmission};
\node[text width=8.5em] (PA2) [startstop_green, left of=OM, yshift = -1cm, xshift = -2.3cm] {\scriptsize Phase alignment algorithm \\ for energy harvesting};

\draw[arrow] (NE.east) -- ++(0.5,0) |- coordinate[pos=0.25](m1)(OM.east);
\node[text width=4em]  [black,left=-0.2 of m1] {\scriptsize Operational mode signal};

\draw[arrow, color = mycolor2] (OM.north) |- coordinate[pos=0.65](m2)(PA1.east);
\node[text width=4em]  [black,above=0 of m2] {\scriptsize {\color{mycolor2} Coverage}};
\draw[arrow, color = mycolor3] (OM.south) |- coordinate[pos=0.65](m3)(PA2.east);
\node[text width=4em]  [black,above=0 of m3] {\scriptsize {\color{mycolor3} Energy harvesting}};

\pgfmathsetmacro{\R}{0.25}
\pgfmathsetmacro{\W}{1.25}
\pgfmathsetmacro{\H}{1.25}
\pgfmathsetmacro{\RISx}{-5.455}
\pgfmathsetmacro{\RISy}{0.6}

\filldraw[fill=apricot, draw=apricot] (\RISx-0.805,\RISy-0.8) rectangle (\RISx+1.1015,\RISy+0.8);
\foreach \i in {-2,...,3} {
    \foreach \j in {-2,...,2} {
    \node [draw, copper, line width=1,trapezium, trapezium left angle=90, trapezium right angle=90, minimum width=\R cm,outer sep=0pt,fill=apricot, rotate = 0] at (\RISx+\i*\W*\R,\RISy+\j*\H*\R) {};
    \draw[copper, line width=1] (\RISx+\i*\W*\R,\RISy+\j*\H*\R+0.5*\R) -- (\RISx+\i*\W*\R,\RISy+\j*\H*\R+0.15*\R);
    \draw[copper, line width=1] (\RISx+\i*\W*\R,\RISy+\j*\H*\R-0.5*\R) -- (\RISx+\i*\W*\R,\RISy+\j*\H*\R-0.15*\R);
    \draw[copper, line width=1] (\RISx+\i*\W*\R-0.25*\R,\RISy+\j*\H*\R-0.15*\R) -- (\RISx+\i*\W*\R+0.25*\R,\RISy+\j*\H*\R-0.15*\R);
    \draw[copper, line width=1] (\RISx+\i*\W*\R-0.25*\R,\RISy+\j*\H*\R+0.15*\R) -- (\RISx+\i*\W*\R+0.25*\R,\RISy+\j*\H*\R+0.15*\R);}}

\node[text width=4em] (PC) [startstop, left of=PA2, yshift = -1.8cm, xshift = -1cm] {\scriptsize Phase \\ controller};

\draw[arrow, color = mycolor3] (PA2.west) -| coordinate[pos=0.8](m4)(PC.north);
\node[text width=4em]  [black,right=0 of m4] {\scriptsize RIS element parameters};

\draw[arrow, color = mycolor2] (PA1.west) -- ++(-0.85,0) -- ++(0,-3.3935);

\node[text width=4em] (EH) [startstop, below of=PC, yshift = -0.2 cm] {\scriptsize Energy \\ harvester};

\draw[dashed] (-6.3,0.6) rectangle (-1.63,5.6);
\node at (-5.51, 5.4) {\scriptsize RIS controller};

\draw[arrow] (EH.east) -| coordinate[pos=0.5](m5)(PA2.south);
\node at (-3,-0.2) {\scriptsize Power measurements};

\draw[arrow] (EH.east) -- coordinate[pos=0.8](m6)(NE.west);

\node[text width=4em] (PS) [startstop, left of=PA2, xshift = -3 cm, yshift = -1.5 cm] {\scriptsize Power \\ supply};

\draw[arrow] (EH.west) -| coordinate[pos=0.8](m7)(PS.south);

% \draw[arrow] (PS.north) |- (-6.3,3.1);
\draw[arrow] (PS.north) |- (-8.284,2.8);

\draw[dashed, thick] (-8.3,-0.55) rectangle (-1.43,5.8);

\node at (-8, 5.6) {\scriptsize RIS};

\node at (-4.2, 0.75) {\scriptsize \rotatebox{90}{RIS elements}};

\end{tikzpicture}}
    \caption{\small The abstract model of the operation of a \gls*{ris} with an energy harvesting module.}
    \label{fig:Harvesting_Structure}
\end{figure}
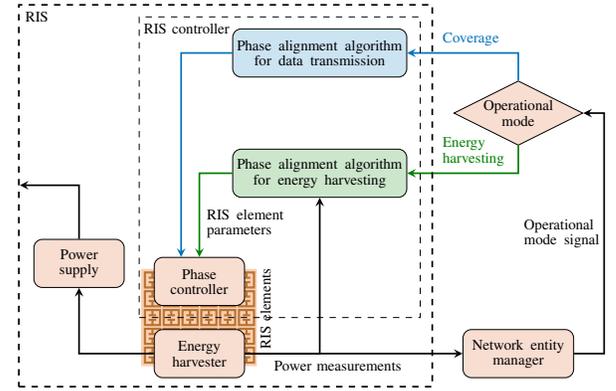

\subsection{Proposed Algorithm}
We consider an ideal scheme, where the \gls*{ris} is capable of adding continuous phase shifts to the incident \gls*{em} waves. 
First, we start with stating Lemma \ref{Lemma_1} that provides a mechanism to find the optimal phase shifts in the case of known \gls*{csi}, then based on that, we develop our proposed algorithm that require no a priori \gls*{csi}.
\begin{lem}\label{Lemma_1}
Let $f:\mathbb{R}^{N}\rightarrow \mathbb{R}$ be the function 
\begin{equation}\label{eq:def_f}
    f{\left(\bm{\vartheta}\right)}\delequal\left|\sum_{n=1}^{N} z_{n}e^{j\vartheta_{n}}\right|^{2},
\end{equation}
where $z_{n}\in \mathbb{C}$ for all $1\leq n\leq N$. The optimal variable $\bm{\vartheta}^{\star}\delequal\left(\vartheta_{1}^{\star}, \dots, \vartheta_{N}^{\star}\right)$ that maximizes $f{\left(\cdot\right)}$ is given by
\begin{equation}
     \vartheta_{n}^{\star}=\vartheta_{0}-\arg\left(z_{n}\right),
\end{equation}
where $\vartheta_{0}\in \mathbb{R}$, and
\begin{equation}
    f{\left(\bm{\vartheta}^{\star}\right)} = \left(\sum_{n=1}^{N}\left|z_{n}\right|\right)^{2}
\end{equation}
is the maximum value of $f{\left(\cdot\right)}$.
\end{lem}

\begin{thm}\label{Theorem_2}
Let $f:\mathbb{R}\rightarrow \mathbb{R}$ be a function of the form
$$f{\left(\vartheta\right)}=\left|z_{0}+ ze^{j\vartheta}\right|^{2},$$ 
where $z_{0}, z\in \mathbb{C}$. Without knowing the explicit values of $z_{0}$ and $z$, the optimal variable $\vartheta^{\star}$ that maximizes $f{\left(\cdot\right)}$ can be computed as
\begin{equation}
      \vartheta^{\star}=\arg{\left(x_{2}+j x_{3}\right)},
\end{equation}
where $\bm{x}\delequal \bm{A}^{-1}\bm{y}$. The matrix $\bm{A}$ is defined as 
\begin{equation} \label{eq:A}
\bm{A} \delequal \begin{bmatrix}
1 & \cos\left(\varphi_{1}\right) &  \sin\left(\varphi_{1}\right)\\
1 & \cos\left(\varphi_{2}\right) &  \sin\left(\varphi_{2}\right)\\
1 & \cos\left(\varphi_{3}\right) &  \sin\left(\varphi_{3}\right)
\end{bmatrix},
\end{equation}
and $\bm{y}\delequal\left[f{\left(\varphi_{1}\right)}, f{\left(\varphi_{2}\right)}, f{\left(\varphi_{3}\right)}\right]^\mathsf{T}$ is the measurement vector. Note that $\varphi_{1}, \varphi_{2}, \varphi_{3}\in \left[0, 2\pi\right)$ are selected such that $\det{\left(\bm{A}\right)}\neq 0$.

 %assuming, $\bm{y}=\left[f{\left(0\right)}, f{\left(\pi/2\right)}, f{\left(\pi\right)}\right]^\mathsf{T}$,

\end{thm}
\begin{proof}
The proof is provided in Appendix \ref{App_T2}.
\end{proof}

In general, Theorem \ref{Theorem_2} allows for the computation of the optimal phase shift using only three measurements, without requiring knowledge of the explicit expression of the function. 
Specifically, for the measurement phases $\varphi_{1}=0$, $\varphi_{2}=\pi/2$, and $\varphi_{3}=\pi$, a simple expression for the optimal phase shift can be obtained as follows:
\begin{equation}\label{eq:contin_phase_update}
\vartheta^{\star}=\arg{\left(y_{1}-y_{3}+j\left(2y_{2}-y_{1}-y_{3}\right)\right)}.
\end{equation}
% Using Theorem \ref{Theorem_2}, we can compute the optimal phase shift with only three measurements rather than knowing the explicit expression of the function. In particular, for the measurement phases $\varphi_{1}=0$, $\varphi_{2}=\pi/2$, and $\varphi_{3}=\pi$, we can obtain a simple expression for the optimal phase shift as
% \begin{equation}\label{eq:contin_phase_update}
%      \vartheta^{\star}=\arg{\left(y_{1}-y_{3}+j\left(2y_{2}-y_{1}-y_{3}\right)\right)}.
% \end{equation}
\RestyleAlgo{ruled}
\SetKwComment{Comment}{/* }{ */}
\begin{algorithm}[t!] \label{Alg1}
\caption{The proposed phase-alignment algorithm for power maximization}
\textbf{Input:} The number of RIS elements $N$ and the number of iterations $M$\\
\textbf{Output:} Near-optimal phase vector $\bm{\vartheta}^{\star}$ that maximizes the received power.\\
\textbf{Initialize:} $\bm{\vartheta}\gets \bm{\vartheta}_{0}$, $m\gets 0$, and $\bm{e}_{n}$ is a vector, where the component $n$ is $1$ and all other components are $0$. 

\While{$m < M$}{
$m \gets m + 1$\;
$n \gets 1$\;
\While{$n \leq N$}{
    $y_{1}\gets$ the measured power for the phase configuration $\bm{\vartheta}$\;
    $y_{2}\gets$ the measured power for the phase configuration $\bm{\vartheta}+\frac{\pi}{2}\bm{e}_{n}$\;
    $y_{3}\gets $ the measured power for the phase configuration $\bm{\vartheta}+\pi\bm{e}_{n}$\;
    $\bm{\vartheta} \gets \bm{\vartheta}+\arg\!\left(y_{1}-y_{3}+j\left(2y_{2}-y_{1}-y_{3}\right)\right)\bm{e}_{n}$\;
    $n \gets n + 1$\;
}}
$\bm{\vartheta}^{\star}\gets \bm{\vartheta}$\;
\end{algorithm}

Algorithm \ref{Alg1} is a sequential, iterative phase update algorithm. At each iteration, adjusting the phase of each element requires measuring the received power for three different phase configurations. The algorithm updates the phase of one element using \eqref{eq:contin_phase_update}, then proceeds to the next element until all $N$ elements have had their phases updated. This process is repeated $M$ times.

\begin{figure}
    \hspace{-15.7mm}
    \scalebox{0.85}{% \definecolor{azure}{rgb}{0.0, 0.5, 1.0}
% \definecolor{awesome}{rgb}{1.0, 0.13, 0.32}
% \definecolor{amber}{rgb}{1.0, 0.49, 0.0}
%\definecolor{darkpastelgreen}{rgb}{0.01, 0.75, 0.24}
\definecolor{mycolor1}{rgb}{0.85,0.325,0.098}
\definecolor{mycolor2}{rgb}{0.00000,0.44700,0.74100}
\definecolor{mycolor3}{rgb}{0.00000,0.49804,0.00000}

\definecolor{apricot}{rgb}{0.98, 0.81, 0.69}
\definecolor{copper}{rgb}{0.72, 0.45, 0.2}
\definecolor{babyblue}{rgb}{0.54, 0.81, 0.94}
\definecolor{gray}{rgb}{0.5, 0.5, 0.5}
\usetikzlibrary{decorations.pathreplacing}
\linespread{0.6}
\tikzstyle{startstop} = [rectangle, rounded corners, minimum width=1.5cm, minimum height=0.8cm,text centered, draw=black, fill=mycolor1!20, inner sep=-0.3ex]
\tikzstyle{startstop_blue} = [rectangle, rounded corners, minimum width=1.5cm, minimum height=0.8cm,text centered, draw=black, fill=mycolor2!20, inner sep=-0.3ex]
\tikzstyle{startstop_green} = [rectangle, rounded corners, minimum width=1.5cm, minimum height=0.8cm,text centered, draw=black, fill=mycolor3!20, inner sep=-0.3ex]

% \tikzstyle{io} = [trapezium, trapezium left angle=70, trapezium right angle=110, minimum width=3cm, minimum height=1cm, text centered, draw=black, fill=blue!30]
%\tikzstyle{process} = [rectangle, minimum width=3cm, minimum height=0.8cm, text centered, draw=black, fill=orange!30, inner sep=-0.3ex]
\tikzstyle{decision} = [diamond, aspect=2,  text centered, draw=black, fill=mycolor1!20, inner sep=-0.3ex]
\tikzstyle{arrow} = [thick,->,>=stealth]

\begin{tikzpicture}[cross/.style={path picture={ 
  \draw[black]
(path picture bounding box.south east) -- (path picture bounding box.north west) (path picture bounding box.south west) -- (path picture bounding box.north east);
}}]

%\node (in1) [io, below of=start] {Input};

%\node (pro1) [process, below of=in1] {Process 1};

\pgfmathsetmacro{\xshift}{4.1}
\pgfmathsetmacro{\xshiftt}{7.1}
\pgfmathsetmacro{\yshiftt}{0.75}

\pgfmathsetmacro{\xshifttt}{10}
\pgfmathsetmacro{\yshifttt}{0.5}

\draw[arrow, color = mycolor2, line width = 1.2] (0,0) -- (1,0.2);
\draw[arrow, color = mycolor2, line width = 1.2] (1,0.2) -- (0.9,0.9);
\draw[arrow, color = mycolor2, line width = 1.2] (0.9,0.9) -- (0.1,0.3);
\node[mycolor2] at (1, 0) {\scalebox{0.9}{$z_{1}$}};
\node[mycolor2] at (1.12, 0.85) {\scalebox{0.9}{$z_{2}$}};
\node[mycolor2] at (0.14, 0.58) {\scalebox{0.9}{$z_{3}$}};

\draw[arrow, color = mycolor2, dashed, line width = 1.2] (\xshift+0,0) -- (\xshift-1.0136,0.1126);
\draw[arrow, color = mycolor2, line width = 1.2] (\xshift-1.0136,0.1126) -- (\xshift-1.1136,0.8126);
\draw[arrow, color = mycolor2, line width = 1.2] (\xshift-1.1136,0.8126) -- (\xshift-1.9136,0.2126);
\draw[arrow, color = mycolor1, dotted, line width = 1.2] (\xshift-1.0136,0.1126) -- (\xshift-1.9136,0.2126);
\node[mycolor2] at (\xshift-0.58, 0.38) {\scalebox{0.9}{$z_{1}e^{j\vartheta^{1}_{1}}$}};
\node[mycolor2] at (\xshift-0.89, 0.78) {\scalebox{0.9}{$z_{2}$}};
\node[mycolor2] at (\xshift-2, 0.4) {\scalebox{0.9}{$z_{3}$}};
\node[mycolor1] at (\xshift-1.8, -0.03) {\scalebox{0.9}{$z_{2}{+}z_{3}$}};

\draw[arrow, color = mycolor2, dashed, line width = 1.2] (\xshiftt+0,\yshiftt+0) -- (\xshiftt-0.6829,\yshiftt- 0.1835);
\draw[arrow, color = mycolor2, line width = 1.2] (\xshiftt-0.6829,\yshiftt- 0.1835) -- (\xshiftt-1.6964,\yshiftt- 0.0709);
\draw[arrow, color = mycolor2, line width = 1.2] (\xshiftt-1.6964,\yshiftt- 0.0709) -- (\xshiftt-2.4964,\yshiftt- 0.6709);
\draw[arrow, color = mycolor1, dotted, line width = 1.2] (\xshiftt-0.6829,\yshiftt- 0.1835) -- (\xshiftt-2.4964,\yshiftt- 0.6709);
\node[mycolor2] at (\xshiftt-1.4, \yshiftt+0.18) {\scalebox{0.9}{$z_{1}e^{j\vartheta^{1}_{1}}$}};
\node[mycolor2] at (\xshiftt-0.42, \yshiftt+0.10) {\scalebox{0.9}{$z_{2}e^{j\vartheta^{1}_{2}}$}};
\node[mycolor2] at (\xshiftt-2.5, \yshiftt-0.39) {\scalebox{0.9}{$z_{3}$}};
\node[mycolor1] at (\xshiftt-2, \yshiftt-0.9) {\scalebox{0.9}{$z_{1}e^{j\vartheta^{1}_{1}}{+}z_{3}$}};

\draw[arrow, color = mycolor2, line width = 1.2] (\xshifttt+0,\yshifttt+0) -- (\xshifttt-0.6829,\yshifttt- 0.1835);
\draw[arrow, color = mycolor2, line width = 1.2] (\xshifttt-0.6829,\yshifttt- 0.1835) -- (\xshifttt-1.6964,\yshifttt- 0.0709);
\draw[arrow, color = mycolor2, dashed, line width = 1.2] (\xshifttt-1.6964,\yshifttt- 0.0709) -- (\xshifttt-2.6956,\yshifttt- 0.1127);
\draw[arrow, color = mycolor1, dotted, line width = 1.2] (\xshifttt+0,\yshifttt+0) -- (\xshifttt-1.6964,\yshifttt- 0.0709);
\node[mycolor2] at (\xshifttt-1.4, \yshifttt-0.36) {\scalebox{0.9}{$z_{1}e^{j\vartheta^{1}_{1}}$}};
\node[mycolor2] at (\xshifttt-0.42, \yshifttt-0.30) {\scalebox{0.9}{$z_{2}e^{j\vartheta^{1}_{2}}$}};
\node[mycolor2] at (\xshifttt-2.52, \yshifttt-0.32) {\scalebox{0.9}{$z_{3}e^{j\vartheta^{1}_{3}}$}};
\node[mycolor1] at (\xshifttt-1.3, \yshifttt+0.22) {\scalebox{0.9}{$z_{1}e^{j\vartheta^{1}_{1}}{+}z_{2}e^{j\vartheta^{1}_{2}}$}};

%% Iteration 2

\pgfmathsetmacro{\xshiftttt}{3}
\pgfmathsetmacro{\yshiftttt}{-3}

\draw[arrow, color = mycolor2, dashed, line width = 1.2] (\xshiftttt+0,\yshiftttt+0) -- (\xshiftttt-1.0108,\yshiftttt - 0.1354);
\draw[arrow, color = mycolor2, line width = 1.2] (\xshiftttt-1.0108,\yshiftttt - 0.1354) -- (\xshiftttt-1.6937,\yshiftttt- 0.3189);
\draw[arrow, color = mycolor2, line width = 1.2] (\xshiftttt-1.6937,\yshiftttt- 0.3189) -- (\xshiftttt-2.6928,\yshiftttt- 0.3606);
\draw[arrow, color = mycolor1, dotted, line width = 1.2] (\xshiftttt-1.0108,\yshiftttt - 0.1354) -- (\xshiftttt-2.6928,\yshiftttt- 0.3606);
\node[mycolor2] at (\xshiftttt-0.72, \yshiftttt+0.15) {\scalebox{0.9}{$z_{1}e^{j\vartheta^{2}_{1}}$}};
\node[mycolor2] at (\xshiftttt-1.4, \yshiftttt-0.5) {\scalebox{0.9}{$z_{2}e^{j\vartheta^{1}_{2}}$}};
\node[mycolor2] at (\xshiftttt-2.52, \yshiftttt-0.6) {\scalebox{0.9}{$z_{3}e^{j\vartheta^{1}_{3}}$}};
\node[mycolor1] at (\xshiftttt-2.25, \yshiftttt-0.01) {\scalebox{0.9}{$z_{2}e^{j\vartheta^{1}_{2}}{+}z_{3}e^{j\vartheta^{1}_{3}}$}};

\pgfmathsetmacro{\xshifttttt}{6.5}
\pgfmathsetmacro{\yshifttttt}{-3.2}

\draw[arrow, color = mycolor2, dashed, line width = 1.2] (\xshifttttt+0,\yshifttttt+0) -- (\xshifttttt-0.7044,\yshifttttt - 0.0621);
\draw[arrow, color = mycolor2, line width = 1.2] (\xshifttttt-0.7044,\yshifttttt - 0.0621) -- (\xshifttttt-1.7152,\yshifttttt- 0.1975);
\draw[arrow, color = mycolor2, line width = 1.2] (\xshifttttt-1.7152,\yshifttttt- 0.1975) -- (\xshifttttt-2.7143,\yshifttttt- 0.2392);
\draw[arrow, color = mycolor1, dotted, line width = 1.2] (\xshifttttt-0.7044,\yshifttttt - 0.0621) -- (\xshifttttt-2.7143,\yshifttttt- 0.2392);
\node[mycolor2] at (\xshifttttt-1.4, \yshifttttt-0.4) {\scalebox{0.9}{$z_{1}e^{j\vartheta^{2}_{1}}$}};
\node[mycolor2] at (\xshifttttt-0.42, \yshifttttt+0.22) {\scalebox{0.9}{$z_{2}e^{j\vartheta^{2}_{2}}$}};
\node[mycolor2] at (\xshifttttt-2.52, \yshifttttt-0.5) {\scalebox{0.9}{$z_{3}e^{j\vartheta^{1}_{3}}$}};
\node[mycolor1] at (\xshifttttt-2.1, \yshifttttt+0.15) {\scalebox{0.9}{$z_{1}e^{j\vartheta^{2}_{1}}{+}z_{3}e^{j\vartheta^{1}_{3}}$}};

\pgfmathsetmacro{\xshiftttttt}{10}
\pgfmathsetmacro{\yshiftttttt}{-3.1}

\draw[arrow, color = mycolor2, line width = 1.2] (\xshiftttttt+0,\yshiftttttt+0) -- (\xshiftttttt-0.7044,\yshiftttttt - 0.0621);
\draw[arrow, color = mycolor2, line width = 1.2] (\xshiftttttt-0.7044,\yshiftttttt - 0.0621) -- (\xshiftttttt-1.7152,\yshiftttttt- 0.1975);
\draw[arrow, color = mycolor2, dashed, line width = 1.2] (\xshiftttttt-1.7152,\yshiftttttt- 0.1975) -- (\xshiftttttt - 2.7086,\yshiftttttt- 0.3118);
\draw[arrow, color = mycolor1, dotted, line width = 1.2] (\xshiftttttt+0,\yshiftttttt+0) -- (\xshiftttttt-1.7152,\yshiftttttt- 0.1975);
\node[mycolor2] at (\xshiftttttt-1.4, \yshiftttttt+0.15) {\scalebox{0.9}{$z_{1}e^{j\vartheta^{2}_{1}}$}};
\node[mycolor2] at (\xshiftttttt-0.42, \yshiftttttt+0.24) {\scalebox{0.9}{$z_{2}e^{j\vartheta^{2}_{2}}$}};
\node[mycolor2] at (\xshiftttttt-2.52, \yshiftttttt+0.05) {\scalebox{0.9}{$z_{3}e^{j\vartheta^{2}_{3}}$}};
\node[mycolor1] at (\xshiftttttt-1.2, \yshiftttttt-0.5) {\scalebox{0.9}{$z_{1}e^{j\vartheta^{2}_{1}}{+}z_{2}e^{j\vartheta^{2}_{2}}$}};

\node[black] at (0.6, -1) {\scalebox{0.9}{Initial Vectors}};
\node[black] at (3.1, -1) {\scalebox{0.9}{Iter 1: Step 1}};
\node[black] at (5.8, -1) {\scalebox{0.9}{Iter 1: Step 2}};
\node[black] at (8.6, -1) {\scalebox{0.9}{Iter 1: Step 3}};

\node[black] at (1.7, -4.6) {\scalebox{0.9}{Iter 2: Step 1}};
\node[black] at (5.2, -4.6) {\scalebox{0.9}{Iter 2: Step 2}};
\node[black] at (8.6, -4.6) {\scalebox{0.9}{Iter 2: Step 3}};

\end{tikzpicture}}
    \caption{Visualization of the proposed algorithm in different steps for a toy example with $N=3$.}
    \label{fig:Visualization_Alg}
\end{figure}
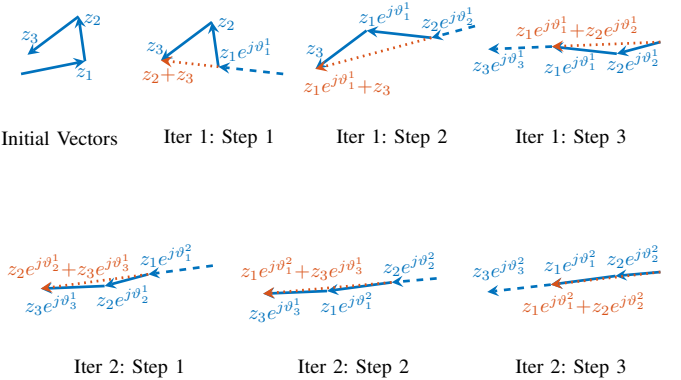

Figure \ref{fig:Visualization_Alg} presents a toy example that demonstrates the different steps of the proposed algorithm. The vector representation shows that initially, the vectors are misaligned, leading to a relatively small amplitude of their sum compared to a scenario where they are aligned in the same direction. The algorithm updates the phase of each vector to match with the sum of the others. As the algorithm progresses and reaches the end of the second iteration, the vectors become almost aligned in the same direction, resulting in a nearly maximum amplitude of their sum.

\begin{thm}\label{Theorem_3}
The proposed Algorithm \ref{Alg1} converges to the maximum value of the function $f{\left(\bm{\vartheta}\right)}=\left|\sum_{n=1}^{N} z_{n}e^{j\vartheta_{n}}\right|^{2}$ as $M\to\infty$. 
\end{thm}
\begin{proof}
The proof is provided in Appendix \ref{App_T3}.
\end{proof}

\subsection{A Random Algorithm}
We use a random algorithm as a benchmark for the proposed algorithm. The basic concept is that, at each step, the algorithm picks a single element of the \gls*{ris} sequentially and assigns a random phase value from a uniform distribution over the interval $[0, 2\pi)$. The new power measurement is then compared to the previous one. If the power increases, the algorithm updates the phase and moves on to the next element. If the measured power decreases, the algorithm maintains the previous phase and proceeds to the next element.

% The continuous random phase-alignment algorithm is described in Algorithm \ref{Alg2}. 

% \begin{algorithm}[t!] \label{Alg2}
% \caption{Continuous random phase-alignment algorithm for power maximization}
% \textbf{Input:} The number of RIS elements $N$, the number of iterations is $M$.\\
% \textbf{Output:} Near-optimal phase vector $\bm{\vartheta}^{\star}$ that maximizes the received power.\\
% \textbf{Initialize:} $\bm{\vartheta}\gets \bm{\vartheta}_{0}$, $m\gets 0$, and $\bm{e}_{n}$ is a vector, where the component $n$ is one and all other components are zero. \\

% $y_{1}\gets$ the measured power for the phase configuration $\bm{\vartheta}$\;
% \While{$m < M$}{
% $m \gets m + 1$\;
% $n \gets 1$\;
% \While{$n \leq N$}{
%     $\alpha \gets \operatorname{rand}\left(0, 2\pi\right)$\;
    
%     $y_{2}\gets$ the measured power for the phase configuration $\bm{\vartheta}+\alpha\bm{e}_{n}$\;
%   \If{$y_{1}<y_{2}$}{
%       $\bm{\vartheta} \gets \bm{\vartheta}+\alpha\bm{e}_{n}$\;
%       $y_{1}\gets y_{2}$\;
%    }
%     $n \gets n + 1$\;
% }}
% $\bm{\vartheta}^{\star}\gets \bm{\vartheta}$\;
% \end{algorithm}

\section{Performance Evaluation}\label{sec:NumResults}
In this section, we evaluate the performance of our proposed algorithm and compare it to the random phase update method. We use a \gls*{ris} with $100$ elements and generate random complex Gaussian distributed values with unit variance for $\left\{z_n\right\}_{n=1}^N$. We conducted Monte-Carlo simulations and compared the results with those obtained from the random phase update method.

Fig. \ref{fig:Mean_vs_measurements} illustrates the normalized achieved power versus the number of power measurements for the proposed and random methods. The proposed method uses three measurements per phase update of each element and converges to its final value after $300$ measurements (i.e., $3N$), while the random one requires ten times more measurements as it has a slow convergence rate. Even in the presence of the noise, at the SNR of $10$\,dB the proposed algorithm reaches the $93$\% of the maximum achievable power \cite{MTavana_Tcom_2023}.

\begin{figure}
    \centering
    \input{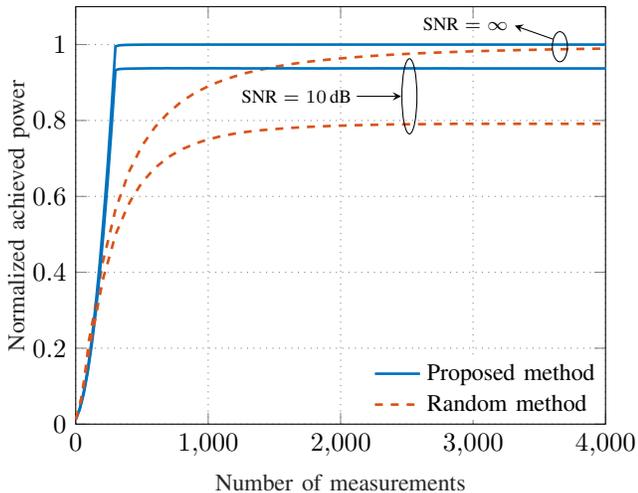}
    \caption{Normalized achieved power versus the number of measurements for the proposed and random methods.}
    \label{fig:Mean_vs_measurements}
\end{figure}

\begin{figure}
    \vspace{-2.25mm}
    \centering
    \input{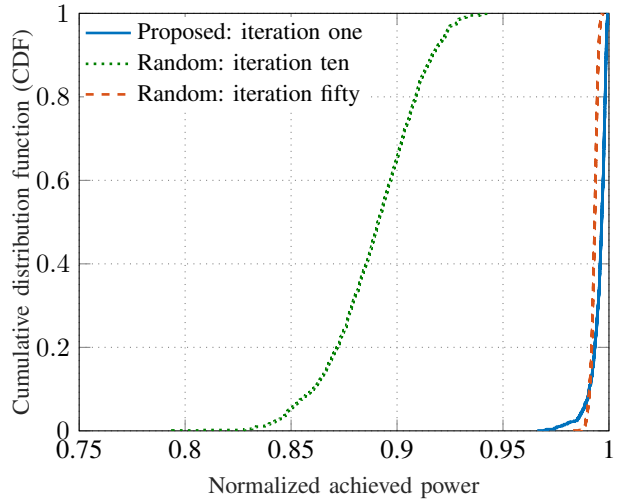}
    \caption{The CDF of the normalized achieved power for the proposed and random algorithms.}
    \label{fig:CDF_proposed}
\end{figure}

Fig. \ref{fig:CDF_proposed} shows the \gls*{cdf} of the relative achieved power for both the proposed and random algorithms, compared to the optimum. We conducted $1000$ simulations with randomly generated channels. The results indicate that the proposed algorithm after the first iteration performs significantly better than the random one after ten iterations. Although the proposed algorithm still performs slightly better than the random one after fifty iterations, the latter appears to be slightly more stable.

\section{Conclusions}\label{sec:Conclusion}
 This paper has presented a method for energy harvesting at a \gls*{ris} from an ambient \gls*{rf} source in the absence of coordination with the source. The objective was to maximize the received power by adjusting the phases of the \gls*{ris} elements based only on power measurements, obtained without having a \gls*{rf} receiver. The proposed sequential algorithm is proved to converge to the optimum and outperformed the random phase update method in terms of convergence rate.
 In future work, we will extend this study to  more general scenarios, including the case with discrete phase shifts.

\appendices
\section{Proof of Theorem \ref{Theorem_2}} \label{App_T2}
Let us define
\begin{equation}\label{eq:x_def_proof}
    \bm{x}\delequal \left[\left|z_{0}\right|^{2}+\left|z\right|^{2}, 2\,\operatorname{Re}\!{\left(z_{0}z^{*}\right)}, 2\,\operatorname{Im}\!{\left(z_{0}z^{*}\right)}\right]^{\mathsf{T}}.
\end{equation}
From Lemma \ref{Lemma_1} and \eqref{eq:x_def_proof}, we have 
\begin{align*}
     \vartheta^{\star}&=\arg{\left(z_{0}\right)}-\arg{\left(z\right)}\\
     &=\arg{\left(z_{0}z^{*}\right)}\\
     &=\arg{\left(\operatorname{Re}\!{\left(z_{0}z^{*}\right)}+j\operatorname{Im}\!{\left(z_{0}z^{*}\right)}\right)}\\
     &=\arg{\left(x_{2}+j x_{3}\right)}.
\end{align*}
Therefore, we can compute the optimal phase shift from $\bm{x}$. We show that one can compute $\bm{x}$ using the received power from three different measurements. By expanding the function $f{\left(\varphi_{l}\right)}$, we get
\begin{align}
    f{\left(\varphi_{l}\right)}&\!=\!\left|z_{0}\right|^{2}\!\!+\! \left|z\right|^{2}\!\!+\!2\operatorname{Re}\left(z_{0}z^{*}\right)\!\cos{\left(\varphi_{l}\right)}\!\!+\!2\operatorname{Im}\left(z_{0}z^{*}\right)\!\sin{\left(\varphi_{l}\right)}\nonumber\\
    &=\bm{a}^\mathsf{T}\bm{x},
\end{align}
where $\bm{a}_{l}\delequal\left[1, \cos{\left(\varphi_{l}\right)},  \sin{\left(\varphi_{l}\right)}\right]^{\mathsf{T}}$, for $1\leq l\leq 3$. 

Assuming $\bm{A}\delequal \left[\bm{a}_{1}, \bm{a}_{2}, \bm{a}_{3}\right]^\mathsf{T}$, and $\bm{y}=\left[f{\left(\varphi_{1}\right)}, f{\left(\varphi_{2}\right)}, f{\left(\varphi_{3}\right)}\right]^{\mathsf{T}}$, we have $\bm{A}\bm{x}=\bm{y}$, or $\bm{x}=\bm{A}^{-1}\bm{y}$ for $\det{\left(\bm{A}\right)}\neq 0$.

\section{Proof of Theorem \ref{Theorem_3}}\label{App_T3}
We denote the phase-shift vector generated by the algorithm at iteration $m$ up to the element $n$ with $\bm{\vartheta}^{k}$, where $k=mN+n$. We prove that for a given $N$, $\lim_{k\to\infty} f{\left(\bm{\vartheta}^{k}\right)}=\max_{\bm{\vartheta}}f{\left(\bm{\vartheta}\right)}$.
%, where $M$ is the maximum value of $f{\left(\cdot\right)}$, and it is defined in Theorem \ref{Theorem_1}.
For any integer $1\leq n\leq N$ and an integer $k\geq 1$, we define 
\begin{equation}
    w_{n}^{k}\delequal \sum_{\substack{i=1\\i\neq n}}^{N} z_{i}e^{j\vartheta^{k}_{i}}.
\end{equation} 
Hence, we have $f{\left(\bm{\vartheta}^{k}\right)} = \left|w^{k}_{n}+z_{n}e^{j\vartheta^{k}_{n}}\right|^2$.
% \begin{equation}
%     f{\left(\bm{\vartheta}^{k}\right)} = \left|w^{k}_{n}+z_{n}e^{j\vartheta^{k}_{n}}\right|^2.
% \end{equation}
At the phase update $k+1$, the algorithm only updates the phase of the element $n$ ($w^{k+1}_{n}=w^{k}_{n}$) as $f{\left(\bm{\vartheta}^{k+1}\right)} = \left|w^{k}_{n}+z_{n}e^{j\vartheta^{k+1}_{n}}\right|^2$.
% \begin{equation}
%     f{\left(\bm{\vartheta}^{k+1}\right)} = \left|w^{k}_{n}+z_{n}e^{j\vartheta^{k+1}_{n}}\right|^2.
% \end{equation}
According to the Theorem \ref{Theorem_2}, we have
\begin{equation}
    f{\left(\bm{\vartheta}^{k+1}\right)} \!=\! \left|\!w^{k}_{n}\!+\!z_{n}e^{j\vartheta^{k+1}_{n}}\!\right|^2\!\geq\! \left|w^{k}_{n}\!+\!z_{n}e^{j\vartheta^{k}_{n}}\right|^2\!=\!f{\left(\bm{\vartheta}^{k}\right)}.
\end{equation}
%therefore, $f{\left(\bm{\vartheta}^{k+1}\right)}\geq f{\left(\bm{\vartheta}^{k}\right)}$.
% \begin{equation}
%     f{\left(\bm{\vartheta}^{k+1}\right)}\geq f{\left(\bm{\vartheta}^{k}\right)}.
% \end{equation}
Therefore, $f{\left(\bm{\vartheta}^{1}\right)}, f{\left(\bm{\vartheta}^{2}\right)}, \dots$ form an increasing sequence. Moreover, according to the Lemma \ref{Lemma_1}, the set $\mathcal{F}\delequal\left\{f{\left(\bm{\vartheta}^{k}\right)}; k\in\mathbb{N}\right\}$ is upper bounded by $\left(\sum_{n=1}^{N}\left|z_{n}\right|\right)^{2}$. Therefore, according to the monotone convergence theorem, we have
\begin{equation}
    \lim_{k\to\infty}f{\left(\bm{\vartheta}^{k}\right)} = \sup{\mathcal{F}}\leq \left(\sum_{n=1}^{N}\left|z_{n}\right|\right)^{2}.
\end{equation}
If we show that $\sup{\mathcal{F}}=\left(\sum_{n=1}^{N}\left|z_{n}\right|\right)^{2}$, the proof is complete. Lets define $\bm{\vartheta}^{\star}\delequal\lim_{k\to\infty}\bm{\vartheta}^{k}$. For the phase shift vector $\bm{\vartheta}^{\star}$, any further phase update will not increase $f{\left(\cdot\right)}$. In other words, for updating element $n$, we should apply Theorem \ref{Theorem_2} to the following function  
\begin{equation}
    f{\left(\bm{\vartheta}^{\star}\right)} = \left|w_{n}^{\star}+z_{n}e^{j\vartheta^{\star}_{n}}\right|^2,
\end{equation}
where $w_{n}^{\star}\delequal \sum_{\substack{i=1\\i\neq n}}^{N} z_{i}e^{j\vartheta^{\star}_{i}}$. Since any further update will not increase the value of $f{\left(\cdot\right)}$, therefore for all $1\leq n\leq N$, we have $\operatorname{Arg}{\left(z_{n}e^{j\vartheta^{\star}_{n}}\right)}=\operatorname{Arg}{\left(w_{n}^{\star}\right)}$. Using Lemma \ref{lemma_2}, we have 
\begin{align}
    \operatorname{Arg}{\left(z_{1}e^{j\vartheta^{\star}_{1}}\right)}=\operatorname{Arg}{\left(z_{2}e^{j\vartheta^{\star}_{2}}\right)}=\dots = \operatorname{Arg}{\left(z_{N}e^{j\vartheta^{\star}_{N}}\right)},
\end{align}
or 
\begin{align}
    \operatorname{Arg}{\left(z_{1}\right)}+\vartheta^{\star}_{1}=\dots= \operatorname{Arg}{\left(z_{N}\right)}+\vartheta^{\star}_{N} = \vartheta_{0} \!\mod 2\pi.
\end{align}
Therefore, we have $\vartheta_{n}^{\star}=\vartheta_{0}-\arg{\left(z_{n}\right)}, \ \text{for all}\ 1\leq n\leq N$,
% \begin{align}
%     \vartheta_{n}^{\star}=\vartheta_{0}-\arg{\left(z_{n}\right)}, \ \text{for all}\ 1\leq n\leq N,
% \end{align}
that are according to the Lemma \ref{Lemma_1}, the phase shifts that maximize $f{\left(\cdot\right)}$.

\begin{lem} \label{lemma_2}
    Assume for each $1\leq n\leq N$, $u_{n}\delequal \sum_{\substack{i=1\\i\neq n}}^{N}z_{i}$, if 
    $\operatorname{Arg}{\left(z_{n}\right)}=\operatorname{Arg}{\left(u_{n}\right)}$ for all $1\leq n\leq N$, then
    $\operatorname{Arg}{\left(z_{1}\right)}=\operatorname{Arg}{\left(z_{2}\right)}=\dots=\operatorname{Arg}{\left(z_{N}\right)}$.
\end{lem}
\begin{proof}
    For $1\leq m, n\leq N$ and $m\neq n$, we have $\operatorname{Arg}{(z_{m})}=\operatorname{Arg}{\left(u_{m}\right)}$ and $\operatorname{Arg}{\left(z_{n}\right)}=\operatorname{Arg}{\left(u_{n}\right)}$, therefore, for some real positive $c_{m}$ and $c_{n}$, we have $z_{m}=c_{m}u_{m}$ and $z_{n}=c_{n}u_{n}$ .
    Assuming $u_{m,n}\delequal \sum_{\substack{i=1\\i\neq m,n}}^{N}z_{k}$, we have
    \begin{align}
        z_{m}&=c_{m}\left(z_{n}+u_{m,n}\right)\label{eq:lemma_a}\\
        z_{n}&=c_{n}\left(z_{m}+u_{m,n}\right)\label{eq:lemma_b}.
    \end{align}
    After some algebraic manipulations, we obtain
    \begin{equation}
        z_{m}=\cfrac{c_{m}\left(1+c_{n}\right)}{c_{n}\left(1+c_{m}\right)}z_{n}.
    \end{equation}
    Hence, $\operatorname{Arg}{\left(z_{m}\right)}=\operatorname{Arg}{\left(z_{n}\right)}$ and the proof is complete.
\end{proof}

\bibliographystyle{IEEEtran}
\bibliography{IEEEabrv,ref}
\end{document}